\documentclass[pdflatex,sn-mathphys-ay]{sn-jnl}
\usepackage{graphicx}
\usepackage{multirow}
\usepackage{amsmath,amssymb,amsfonts}
\usepackage{amsthm}
\usepackage{mathrsfs}
\usepackage[title]{appendix}
\usepackage{xcolor}
\usepackage{textcomp}
\usepackage{manyfoot}
\usepackage{booktabs}
\usepackage{algorithm}
\usepackage{algorithmicx}
\usepackage{algpseudocode}
\usepackage{listings}
\usepackage[utf8]{inputenc}
\theoremstyle{thmstyleone}
\newtheorem{theorem}{Theorem}
\newtheorem{proposition}[theorem]{Proposition}

\usepackage{tikz}
\usetikzlibrary{arrows.meta}

\theoremstyle{thmstylethree}
\newtheorem{definition}[theorem]{Definition}
\theoremstyle{thmstyletwo}

\begin{document}
\title{An Integrated Framework for Rule-Based Data Editing}
\author*[1]{\fnm{Pierpaolo} \sur{Massoli}}
\email{pimassol@istat.it}
\affil[1]{%
\orgdiv{Directorate for Methodology and Statistical
Process Design (DCME)},
\orgname{Italian National Institute of Statistics (ISTAT)},
\orgaddress{\street{Via Cesare Balbo 16},
\city{Rome}, \postcode{00184}, \country{Italy}}}
%

%
\abstract{
Statistical data editing must restore consistency while limiting
changes to observed data. This paper links questionnaire routing,
semantic edits, diagnosis, and correction verification through a
common specification. The framework applies to dependencies admitting
a normalized acyclic routing representation. Routing metadata generate
Horn clauses requiring selected questions to be answered and skipped
questions to remain blank. Together with supplied semantic edits and
explicit domain and auxiliary conditions, these requirements
characterize admissibility without generating all implied rules.
Indexed requirements link violations to their sources and identify
inputs whose modification can affect them. A weighted violation count
guides correction: with complete checks and positive weights, it is
zero exactly for admissible records. Returned corrections can
therefore be verified independently of the search procedure. Minimum
modification cost is addressed separately through sufficient
conditions for a combined penalty objective and certificates obtained
when an output attains a valid cost lower bound. A real household survey
application illustrates a practical use of the proposed framework. 
The workflow supports traceable
consistency restoration; the statistical accuracy of replacement
values and their effects on survey estimates require separate
assessment.
}

\maketitle
%


\section{Introduction}
Statistical data editing in official statistics must identify
inconsistencies, restore compliance with prescribed requirements, and
limit unnecessary changes to observed data. In questionnaire-based
surveys, routing rules determine which questions apply, while semantic
edits specify relationships among recorded values. Changing a response
may activate skipped questions, invalidate existing answers, or alter
a semantic check. Correction must therefore account for consequences
across the complete record \citep{Fagan1988,Fellegi1976,deWaal2011}.
\par
The Fellegi--Holt paradigm seeks a minimum-weight set of fields whose
modification permits a consistent completion while retaining the other
observed values \citep{Fellegi1976,deWaal2011}. Identifying failed
edits alone does not solve this problem: intersecting every currently
violated rule need not permit jointly consistent replacements, because
modifications may activate further requirements.
\par
Exact approaches include implicit-edit generation and set covering
\citep{Fellegi1976,Garfinkel1986,Winkler1997}, cutting planes
\citep{Garfinkel1988,RagsdaleMcKeown1996}, variable elimination
\citep{deWaalQuere2003,deWaal2005}, vertex generation
\citep{Sande1978,deWaal2003b}, and mathematical programming
\citep{RieraSalazar2007,deJonge2014}. They provide exact solutions
under their respective assumptions, and some avoid generating all
implied edits in advance. Nevertheless, minimum-weight field selection
is NP-hard for categorical data with general edits
\citep{Garfinkel1986}, and intermediate representations can become
large \citep{deWaal2003a,deWaal2003b,deWaal2011}. Computational
demands depend on variable domains, dependency structure, and
constraint interactions. This motivates distinguishing the search for
a correction from the guarantees established for its output:
admissibility can be checked directly, and minimum cost can be
certified when the attained cost matches a valid lower bound.
\par
A related literature concerns the representation of editing
requirements. Graph-based descriptions support skip-pattern analysis
\citep{Fagan1988}, questionnaire modelling \citep{Massoli2008}, and
compatibility-rule generation \citep{LoCastroMassoli2021}. Editing
systems support constraint specification, checking, and processing
\citep{Concord2004,Pannekoek2005,deJonge2025,vanDerLooDeJonge2021}.
Database-repair research studies consistency restoration under
declared constraints \citep{BohannonEtAl2005,Bertossi2011}, within
data quality management \citep{FanGeerts2012}.
\par
This study provides an integrated framework linking questionnaire
routing metadata and supplied semantic edits to diagnosis and
verifiable correction. Its distinctive contribution is a common,
traceable specification connecting rule generation, conflict
diagnosis, assessment of candidate changes, and verification of
corrected records. A normalized directed acyclic graph generates
indexed routing requirements expressed as Horn clauses
\citep{Horn1951,Lloyd1987,Miller2006}. Together with supplied semantic
edits and explicit domain and auxiliary conditions, these requirements
characterize admissibility without generating the full set of implied
rules.
\par
For official statistical production, this integration offers three
practical benefits. First, deriving routing requirements from
questionnaire metadata reduces the need to maintain a separate
routing-rule specification. Second, each diagnosed conflict retains
a link to its generating dependency and to the inputs whose
modification can affect it, supporting investigation and documentation.
Third, corrected records are checked against the same requirements,
including those activated by changes to earlier responses.  
Together, these features help methodologists manage the editing
process, document the rationale for interventions, and audit
returned corrections within a single integrated workflow. This
requires an initial investment in questionnaire modelling and
metadata specification, which supports automation across these
stages. The formal results establish the conditions under which
the shared specification characterizes record admissibility.
\par
The diagnostics also define a weighted violation penalty, called the
violation Hamiltonian. With complete checks and strictly positive
weights, zero energy is equivalent to admissibility within the
declared domain. This provides a correction target and allows
returned records to be verified independently of the search
procedure. Minimum modification cost remains a separate objective:
the framework gives sufficient conditions for a combined penalty
formulation and certificates for individual outputs whose costs
attain valid lower bounds. Candidate checking requires neither prior
minimum-cost localization nor graph partitioning.
\par
The results apply to specifications satisfying the stated
representation and semantic assumptions. Questionnaire routing is
the main motivation; registers and administrative data are also
covered when their dependencies admit the specified normalized
acyclic representation. A household survey application illustrates
the construction, correction verification, and a retrospective
minimum-cost audit. It does not compare correction algorithms or
establish computational superiority. Admissibility and minimum
intervention remain distinct from the statistical accuracy of
replacement values and their effects on survey estimates.
\par
Section~\ref{sec. theoretical_background} introduces the background
concepts through a questionnaire example.
Section~\ref{sec:mathematical-framework} develops the framework, and
Section~\ref{sec:empirical-validation} presents the application.
Sections~\ref{sec:discussion} and~\ref{sec:conclusion} discuss the
findings and conclude. Appendix~\ref{app:formal} contains the formal
construction and proofs; Appendix~\ref{app:checks} gives
implementation checks and operational details.


\section{Theoretical Background}
\label{sec. theoretical_background}

This section introduces the graph-theoretical and logical concepts
used to represent structural dependencies and record consistency,
and distinguishes the main tasks in rule-based editing.

\subsection{Directed Acyclic Graphs}
A finite directed graph is a pair $G=(V,E)$, where
$V=\{v_1,\ldots,v_n\}$ is the vertex set and
$E\subseteq V\times V$ is the set of directed edges.
A directed path follows successive edges, and a directed cycle
returns to its starting vertex.
A \textit{directed acyclic graph} (DAG) contains no directed
cycles and therefore admits a \textit{topological ordering}:
an ordering of its vertices in which every edge points from
an earlier to a later vertex \citep{Diestel2005}.
This ordering need not be unique.
\par
In a routing graph, vertices represent variables and edges
represent possible transitions between them.
A \textit{modality} is a category or value in a variable's
prescribed domain.
Several modalities may select the same transition, so graph
adjacency alone does not identify the response associated
with an edge.

\subsection{Incidence Matrix Representation}
For a directed graph without self-loops, with $n$ vertices
and $m$ indexed edges, the \textit{incidence matrix}
$B=(b_{ij})\in\{-1,0,+1\}^{n\times m}$ is defined using
the following sign convention:
\begin{equation}
b_{ij} =
\left\{
\begin{array}{rl}
+1, & \mbox{if edge } e_j \mbox{ leaves vertex } v_i,\\
-1, & \mbox{if edge } e_j \mbox{ enters vertex } v_i,\\
0,  & \mbox{otherwise.}
\end{array}
\right.
\end{equation}
Each column identifies the source and target of one edge.
Parallel edges in a directed multigraph retain distinct indexed
columns, even when their endpoints coincide.
Transposition gives one row per indexed edge.
Ordinary incidence records endpoints but does not explicitly
encode the response modalities selecting a routing transition
or the variables made inactive by that decision.
Section~\ref{sec:mathematical-framework} introduces an expanded
representation retaining this information, with one indexed row
for each modality-specific transition.

\subsection{Horn Logic}
\label{sec:horn_logic}

An \textit{atom} represents an elementary proposition,
and a \textit{literal} is an atom or its negation.
A \textit{Horn clause} is a disjunction of literals containing
at most one positive literal \citep{Horn1951,Lloyd1987}.
It can be written in implication form as
\begin{equation}
a_1 \wedge \cdots \wedge a_k \rightarrow b,
\end{equation}
where the $a_i$ are atoms and $b$ is either an atom or $\bot$.
The symbols $\top$ and $\bot$ denote truth and falsity,
respectively.
Atoms may represent variable--value assignments, such as
$V_i=a$, or activity states.
The conjunction on the left is the \textit{antecedent};
the right-hand side is the \textit{consequent}.
An empty antecedent with an atomic consequent gives
an unconditional fact, $\top\rightarrow b$.
A consequent $\bot$ gives an \textit{integrity constraint},
forbidding the simultaneous truth of all antecedent atoms. 
A \textit{Horn theory} $\mathcal{H}$ is a finite set
of Horn clauses.
An interpretation assigns a truth value to each atom.
An interpretation $x$ is a model of the theory,
written $x\models\mathcal{H}$, if it satisfies every clause.
The notation $\mathcal{H}\models\varphi$ means that every
model of $\mathcal{H}$ also satisfies $\varphi$.
\par
Forward chaining derives consequences from supplied facts under the
Horn requirements. Record validation instead evaluates the
interpretation fixed by the observed coordinates, including their
modal and activity states. Appendix~\ref{subsec:horn-theory} gives the
inference properties used below.

\subsection{Rule-Based Data Editing}
An edit rule specifies a relationship that a valid record
must satisfy.
A collection of edit rules defines an admissible data space
relative to the prescribed domains and requirements.
These may concern routing and structural missingness,
or semantic relationships among recorded values. 
Consistency checking determines whether a record satisfies
the applicable rules.
Diagnostic localization identifies failed requirements
and the variables involved.
Under the Fellegi--Holt paradigm, optimal error localization
selects a set of fields with minimum total reliability weight
whose modification permits a consistent completion,
while retaining the remaining observed values.
Imputation assigns replacement values satisfying
the joint requirements \citep{Fellegi1976,deWaal2011}. 
Covering the initially failed requirements is necessary but does not
ensure that replacements satisfy all requirements, including those
activated by the modifications.

\subsection{From Routing and Semantic Metadata to Edit Rules}
\label{sec:illustrative_questionnaire}

The following example illustrates the use of routing instructions
and explicitly specified semantic constraints, translated into edit
rules, to determine record admissibility.
Consider a questionnaire collecting information on housing support
and expenditure during the same reference period. Question $Q_1$
asks whether the household received a housing allowance, with
possible answers Yes and No. A Yes answer leads to $Q_2$, which
asks for the total allowance received in euros, and then to $Q_3$.
A No answer skips $Q_2$ and leads directly to $Q_3$, which asks
for the total housing expenditure in euros and terminates this
questionnaire fragment.
For this illustrative questionnaire, amounts are recorded as whole
euros in a prescribed finite domain $\{0,\ldots,M\}$, where
$M\geq 1200$ to accommodate the records considered below.
The questionnaire also specifies two semantic edits: a reported
allowance must be strictly positive and must not exceed the housing
expenditure reported for the same period. These are assumptions of
the example, rather than general properties of housing allowances.
\par
Figure~\ref{fig:example-routing} represents the modal routing
transitions and the associated skip requirement. The two admissible
paths are $Q_1,Q_2,Q_3$ and $Q_1,Q_3$. Question $Q_1$ is the
required root, and $Q_3$ terminates both paths.
\begin{figure}[htbp]
\centering
\begin{tikzpicture}[
  >=Stealth,
  question/.style={
    draw,
    circle,
    minimum size=9mm,
    inner sep=2pt,
    font=\small
  },
  route/.style={
    ->,
    line width=0.6pt
  },
  route label/.style={
    font=\footnotesize,
    align=center,
    fill=white,
    inner sep=2pt
  }
]
\node[question] (q1) at (0,0) {$Q_1$};
\node[question] (q2) at (2.2,0) {$Q_2$};
\node[question] (q3) at (5.4,0) {$Q_3$};

\draw[route] (q1) --
  node[route label,above] {Yes} (q2);

\draw[route] (q2) --
  node[route label,above]
  {$a\in\{0,\ldots,M\}$} (q3);

\draw[route] (q1.south)
  .. controls +(0,-1.4) and +(0,-1.4) ..
  node[route label,pos=0.5,below]
  {No\\$Q_2$ must be blank} (q3.south);
\end{tikzpicture}
\caption{Questionnaire routing DAG.}
\label{fig:example-routing}
\end{figure}
The edge labels retain the selecting answers and skip requirement
omitted by ordinary incidence. The label $a\in\{0,\ldots,M\}$
abbreviates separately indexed events for the permitted amounts, all
selecting the same edge. Section~\ref{sec:mathematical-framework}
retains this information in an expanded representation.
\par
To express the corresponding requirements logically, let $P_i$
denote that $Q_i$ contains a value in its prescribed domain,
and let $B_i$ denote that $Q_i$ is structurally inactive and
recorded as \texttt{NULL}. These atoms are evaluated from the
record: they are not independent labels that can be freely
assigned. In particular, a question cannot be both filled and
structurally inactive. Table~\ref{tab:example_notation} summarizes
the notation used below.
\begin{table}[htbp]
\centering
\small
\caption{Notation used in the illustrative example.}
\label{tab:example_notation}

\begin{tabular}{@{}lp{0.76\linewidth}@{}}
\toprule
Symbol & Meaning \\
\midrule
$P_i$ & Question $Q_i$ contains a value in its prescribed domain. \\
$B_i$ & Question $Q_i$ is structurally inactive and recorded as
        \texttt{NULL}. \\
$\to$ & If the premise holds, the conclusion must hold. \\
$\bot$ & Contradiction: the condition implying it is forbidden. \\
\texttt{NULL} & For a questionnaire answer, structural absence,
               distinct from numerical zero. \\
\bottomrule
\end{tabular}
\end{table}
\par
The routing instructions generate the Horn clauses
\begingroup
\begin{align}
(Q_1=\mathrm{Yes}) &\longrightarrow P_2,
  \label{eq:example_yes}\\
(Q_1=\mathrm{No}) &\longrightarrow P_3,
  \label{eq:example_no_target}\\
(Q_1=\mathrm{No}) &\longrightarrow B_2,
  \label{eq:example_no_skip}\\
(Q_2=a) &\longrightarrow P_3,
  \qquad a\in\{0,\ldots,M\}.
  \label{eq:example_amount}
\end{align}
\endgroup
Each clause links a checkable requirement to its originating
instruction. Equation~\eqref{eq:example_no_skip} requires a structural
blank, not a numerical zero. Equation~\eqref{eq:example_yes} requires
an answer to $Q_2$ but does not establish its semantic validity.
The semantic requirements can be incorporated through a
deterministically computed outcome $E$. When $Q_1=\mathrm{Yes}$
and both amounts contain valid domain values, define
\[
E=
\begin{cases}
1, & 0<Q_2\leq Q_3,\\
\mathrm{NULL}, & \text{otherwise}.
\end{cases}
\]
Set $E=2$ when the edit is inactive because $Q_1=\mathrm{No}$,
or cannot be evaluated because a required input is unavailable.
Here, $E=\mathrm{NULL}$ encodes a detected semantic violation;
it does not denote a skipped questionnaire answer.
The accepted outcomes are 1 and 2; an evaluable violation is
excluded by the Horn integrity constraint
\[
(E=\mathrm{NULL})\longrightarrow\bot.
\]
Outcome 2 does not establish that the amounts are consistent.
If $Q_1=\mathrm{Yes}$ and $Q_2$ is missing, for example, the
semantic comparison is unevaluable but
Equation~\eqref{eq:example_yes} still reports a routing failure.
Moreover, $E$ is recomputed from the answers after every change;
it cannot be independently changed to make a record acceptable.
Three records illustrate the distinction between the checks.
\begin{center}
\begin{tabular}{@{}llll@{}}
\toprule
$Q_1$ & $Q_2$ & $Q_3$ & Assessment \\
\midrule
Yes & 1200 & 1000 & Semantic edit violated \\
No & 1200 & 1000 & Skipped question filled \\
No & \texttt{NULL} & 1000 & Admissible \\
\bottomrule
\end{tabular}
\end{center}
The first record violates the semantic comparison, the second fills a
skipped question, and the third is admissible. Correction requires
checking the complete modified record. Changing Yes to No in the first
record deactivates the comparison but requires $Q_2$ to be blanked as
well. Alternatively, changing the allowance to a positive amount not
exceeding 1000 restores consistency while retaining the path. Both
corrections require verification; choosing between them requires an
intervention criterion.


\section{Mathematical Framework}
\label{sec:mathematical-framework}

The framework translates a specified routing structure and semantic
edits into diagnostics that can be evaluated on successive candidate
corrections. This section presents the construction and its
operational consequences; Appendix~\ref{app:formal} gives the complete
definitions, statements, and proofs. 

\subsection{From Metadata to Record Validation}
\label{subsec:main-specification}

The framework focuses on categorical variables and structural
presence or absence requirements. For quantitative items, routing
diagnostics identify missing required values or values in skipped
items, without determining numerical replacements. Arithmetic
semantic edits are supported for quantitative variables whose
permitted values can be represented as a finite set of categorical
modalities, preserving the numerical meaning required to evaluate
the edits. General quantitative imputation is outside the scope
of this study. 
Normalization resolves compound
routing constructs and introduces auxiliaries where needed, following
\citet{Massoli2008}. The resulting graph is finite and acyclic, with
one distinguished root. Each variable--modality pair selects one
successor, possibly the terminal state, and identifies variables
required to be inactive. These requirements must be coherent with
permitted paths, and finite domains must retain all distinctions used
by semantic predicates.
\par
The expanded transposed incidence representation indexes these
transitions by source and modality. Each row records the successor and
skipped variables and generates requirements for the successor to be
filled and skipped variables to be blank. Distinct events retain their
indices even when endpoints coincide; terminal events may generate no
clause. Canonicity refers to the fixed normalization and domains. A
coherent routing specification yields jointly satisfiable routing
requirements; supplied semantic edits may introduce additional
conflicts.
\par
Atoms are evaluated from record coordinates: an active variable has
one permitted value, while $\mathrm{NULL}$ denotes structural
inactivity. The root must be active, and each active nonroot variable
must have exactly one active modal predecessor selecting it. These
conditions define $\mathcal X_G$. Within this domain, the Horn
requirements characterize admissible paths
(Theorem~\ref{thm:soundness-completeness}). On unrestricted inputs,
the additional conditions must also be checked: an entirely blank
record otherwise satisfies every routing implication vacuously.
\par
Semantic predicates are evaluated from each candidate using the
outcome convention illustrated in
Section~\ref{sec:illustrative_questionnaire}. Their outcomes cannot be
assigned independently, and auxiliary dependencies require an acyclic
evaluation order. Horn constraints reject forbidden outcomes without
requiring the arithmetic predicates themselves to be Horn formulae.
Together, the domain, routing, and semantic checks characterize joint
admissibility (Theorem~\ref{thm:semantic-reduction}); validation
requires no additional implied rules
(Theorem~\ref{thm:closure-free-consistency}).

\subsection{From Diagnostics to Candidate Corrections}
\label{subsec:main-diagnostics}

A diagnostic identifies a failed requirement and the editable inputs
that can affect it, including its activating answer and the underlying
inputs of derived coordinates. Activated requirements are compared
directly with observed states, using the indexed representation or its
matrix form (Appendix~\ref{subsec:canonical-localization}). This
requires neither minimum-cost field selection nor enumeration of the
admissible record set.
\par
A conflict with editable scope $C_r$ requires every admissible
correction to change at least one input in $C_r$. Covering all such
scopes is necessary but does not guarantee jointly valid replacements:
an upstream change can activate questions or require downstream
blanking. Appendix~\ref{subsec:localization-correction} develops these
coverage conditions and their use in cost lower bounds. Each proposed
change is evaluated with its routing consequences and recomputed
semantic outcomes. A diagnostic scan is linear in inspected
occurrences when individual checks have bounded cost, plus
predicate-evaluation costs; this bounds candidate checking, not
correction search. No preliminary graph partition is required.
\par
For an observed record $x$, let $\mathcal Y(x)$ be a nonempty finite
candidate domain respecting permitted values, fixed fields, modal
encoding, and the conditions defining $\mathcal X_G$. Candidates must
satisfy auxiliary coherence and the declared restrictions on external
or mode-locked information. Missing external information cannot be
freely replaced to remove a violation.
Appendix~\ref{subsec:hamiltonian-correction} specifies how auxiliary
coordinates are treated during search and verification.

\subsection{An Energetic Target for Admissible Correction}
\label{subsec:main-energy}

Assign each candidate a weighted penalty for its outstanding
violations. Let $N_R(y)$, $N_L(y)$, and $N_A(y)$ count routing
violations on ordinary variables, semantic violations, and auxiliary
violations. The last class includes routing failures involving
auxiliaries and applicable disagreements with their defining rules.
Counts concern diagnostic occurrences and must cover every required
check. 
Representing constraint violations through nonnegative energy
penalties follows an established approach in combinatorial
optimization \citep{Lucas2014}. 
The violation Hamiltonian is
\begin{equation}
H_{\mathrm{viol}}(y)=w_RN_R(y)+w_LN_L(y)+w_AN_A(y),
\qquad w_R,w_L,w_A>0.
\label{eq:main-violation-energy}
\end{equation}
With complete diagnostics and strictly positive weights,
$H_{\mathrm{viol}}(y)=0$ if and only if $y$ is admissible under the
declared domain and auxiliary conditions
(Proposition~\ref{prop:hamiltonian-zero-penalty}). No prescribed
failure may escape all counts. The weights rank inconsistent
candidates without changing the zero set. If candidates outside
$\mathcal X_G$ are inspected, root and predecessor failures must also
be excluded or diagnosed.
\par
A verified zero-energy correction certifies feasibility and attains
the global minimum of violation energy by nonnegativity, independently
of how search produced it. It does not certify minimum modification
cost. Conversely, failure to find a zero-energy correction does not
establish infeasibility.

\subsection{Minimum Intervention and Output Certificates}
\label{subsec:main-cost}

Zero violation energy alone does not distinguish corrections by the
amount of intervention. Let $\mathcal F(x)$ denote the admissible
subset of $\mathcal Y(x)$ and let
\begin{equation}
C(y\mid x)=\sum_{j\in J}c_j(x_j,y_j),
\qquad c_j(a,b)\geq0,\quad c_j(a,a)=0,
\label{eq:main-cost}
\end{equation}
with finite costs and a declared set $J$ of charged coordinates. Costs
may distinguish filling, blanking, and replacement. They compare the
final candidate with the original record, including induced changes on
charged coordinates; accumulated search-move costs need not coincide
with this quantity.
\par
When $\mathcal F(x)$ is nonempty, minimum intervention is
$C^\star(x)=\min_{y\in\mathcal F(x)}C(y\mid x)$.
A scalar objective combining both components is
\begin{equation}
H_\lambda(y\mid x)=H_{\mathrm{viol}}(y)+\lambda C(y\mid x),
\qquad \lambda>0.
\label{eq:main-penalized-energy}
\end{equation}
Positive coefficients alone do not guarantee admissible minimizers: a
cheap inconsistent record could have lower combined energy. Its global
minimizers are minimum-cost admissible corrections if
$\lambda C(y_0\mid x)<\delta$ for an available feasible correction
$y_0$ and a positive lower bound $\delta$ on every infeasible candidate's
violation energy (Theorem~\ref{thm:hamiltonian-exact-penalty}). For
the integer diagnostic counts above, one may use
$\delta=\min\{w_R,w_L,w_A\}$. Lexicographic minimization of violation
energy first and modification cost second expresses the same priority
without this scalar weight condition. Neither formulation guarantees
that a search finds a global minimizer.
\par
An individual correction can nevertheless receive a minimum-cost
certificate. If $L(x)$ is a valid lower bound on every feasible
correction's cost, an admissible output $\widehat y$ satisfies
\begin{equation}
L(x)\leq C^\star(x)\leq C(\widehat y\mid x).
\label{eq:main-certificate}
\end{equation}
Equality certifies optimality. In
Section~\ref{sec:illustrative_questionnaire}, take
$(Q_1,Q_2,Q_3)=(\mathrm{Yes},\mathrm{NULL},1000)$.
If all three answers are charged and every change costs at least one,
the failed requirement for $Q_2$ gives $L(x)=1$. Filling $Q_2$ with
$500$ at cost one produces an admissible record and attains this
bound. Lower bounds must use the same cost accounting, including zero
weights for uncharged inputs. The application certifies minimum cost
for a subset of outputs in this way.
%

\section{Empirical Validation}
\label{sec:empirical-validation}

The empirical study applies the framework to an anonymized extract
from a real household survey. It examines the correspondence between
metadata and generated rules, the integration of routing and semantic
diagnostics, and the admissibility of returned corrections. A
controlled synthetic experiment checks agreement between validation
and localization, while a retrospective audit identifies corrections
whose minimum modification cost can be certified.

\subsection{Case Study and Validation Design}
\label{subsec:empirical-design}

The input contains $15\,000$ household records and $332$ columns,
including a record identifier and a five-category territorial
partition.
It was selected from a working dataset of $21\,470$ records
by retaining all $4\,842$ records classified as consistent
by the extended routing check and sampling $10\,158$
inconsistent records proportionally across province-code
groups, with seed $123$.
The original row order was retained.
Since selection depends on the initial consistency classification,
the reported violation rates describe this validation sample,
rather than population error prevalence. 
\par
Four metadata files specify routing, auxiliary-variable
construction, logical edits, and item mappings.
Three external sources supply information for exogenous
auxiliaries.
Two cover all selected households; the third covers $14\,988$,
leaving $12$ unmatched identifiers already absent from
the source data.
These gaps are retained during preparation. 
Preparation constructs auxiliary variables and logical-group
representations and aligns the resulting $255$ columns
with the canonical routing graph.
This representation includes $30$ OR/XOR group nodes
and three exogenous auxiliary variables.
Two additional nodes subsequently encode semantic edits.
The $332$ input columns, $255$ routing variables,
and $257$ extended variables therefore describe successive
representations of the data.

\subsection{Canonical Representation and Rule Generation}
\label{subsec:empirical-canonicity}

The normalized routing graph contains $255$ vertices
and $332$ distinct directed edges.
It is acyclic, has one source node and no isolated nodes.
The stored modal incidence representation contains $572$ rows.
Each has exactly one source, at least one destination,
a unique modal identifier, and at least one generated
implication.
All six implemented canonical tests pass.
\par
The generator writes $572$ compatibility rules, expanding
into $1\,008$ distinct atomic implications:
$572$ require a target to be filled and $436$ require
a skipped variable to be blank.
Parsing the saved rule file reproduces exactly this atomic
set, with no unmatched implications.
Table~\ref{tab:empirical-structure-horn} summarizes these
results and the additional implementation checks in
Appendix~\ref{subsec:empirical-minimality}.
\begin{table}[htbp]
\centering
\caption{Canonical structure and logical implementation checks.}
\label{tab:empirical-structure-horn}
\begin{tabular}{lr}
\toprule
Quantity or check & Result \\
\midrule
Routing vertices & $255$ \\
Distinct directed edges & $332$ \\
Stored nonterminal modal rows & $572$ \\
Canonical structural tests passed & $6/6$ \\
Atomic target-filled implications & $572$ \\
Atomic skip-blank implications & $436$ \\
Total elementary implications & $1\,008$ \\
Equality with parsed compatibility rules & Exact \\
Deletion tests returning non-derivability & $1\,008/1\,008$ \\
Operational implications, including bridges & $1\,580$ \\
Additional implications after materialization & $0$ \\
Singleton inference comparisons & $977$ \\
Random multiple-seed comparisons & $5\,000$ \\
Inference disagreements & $0$ \\
\bottomrule
\end{tabular}
\end{table}
These results verify agreement between the stored metadata and
generated rules, conditional on the supplied specification, as
described in Theorem~\ref{thm:canonical-horn-representation}.
Appendix~\ref{subsec:empirical-patterns} reports the routing-pattern
counts and supplementary structural analysis.

\subsection{Routing and Semantic Violation Localization}
\label{subsec:routing-localization}

On the $15\,000$ observed records, routing validation reports
$16\,093$ incompatibilities involving $10\,153$ distinct
households and $63$ variables.
Localization reuses the saved routing check and returns
the same incompatibilities and record count.
This verifies agreement between workflow stages and preserves
the association between each failed requirement
and its affected variables. 
Two logical-edit nodes are then added while preserving acyclicity and
data alignment; Appendix~\ref{subsec:empirical-minimality} gives the
implementation counts.
The first edit checks distinctness of nonblank entries
in a six-field list and requires its first entry when
the list is active.
The second requires six associated percentage fields
to sum to $100$, treating blank components as zero
and using the recorded numerical values.
For both predicates, the all-blank configuration is explicitly
accepted as outcome $2$, while an accepted active configuration
yields outcome $1$.
Records satisfying neither condition yield \texttt{NA},
represented below by $\mathrm{NULL}$.
These conventions specify how the two supplied edits
treat blank inputs.
\begin{table}[htbp]
\centering
\caption{Evaluation of the two logical edits on observed records.}
\label{tab:semantic-elog-outcomes}
\begin{tabular}{lrrr}
\toprule
Edit node & Outcome $1$ & Outcome $2$ & $\mathrm{NULL}$ \\
\midrule
\texttt{LEDIT\_1} & $10\,864$ & $4\,118$ & $18$ \\
\texttt{LEDIT\_2} & $10\,125$ & $4\,118$ & $757$ \\
\bottomrule
\end{tabular}
\end{table}
The $775$ forbidden outcomes in
Table~\ref{tab:semantic-elog-outcomes} correspond exactly
to the semantic violations in the extended report.
Combined with $16\,093$ routing violations, they give
$16\,868$ incompatibilities on $10\,158$ records,
leaving $4\,842$ initially consistent records.
Only five additional households are classified as inconsistent
relative to routing alone, because routing and semantic
conflicts overlap.
Semantic counts refer to record--edit pairs rather than
distinct households.
Their agreement with the report verifies diagnostic
integration, rather than independent semantic ground truth. 
\par
A complementary synthetic experiment generates $5\,000$
artificial records from the same routing metadata,
with seed $123$.
Routing validation accepts all of them with zero violations.
Using seed $124$, the procedure then attempts alterations
in $2\,500$ records.
Alterations not detected by the checker are reverted,
leaving $2\,265$ records with retained detectable alterations
and $2\,735$ records classified as admissible.
The resulting dataset contains $2\,683$ routing incompatibilities.
\begin{table}[htbp]
\centering
\caption{Record-level localization on the filtered synthetic test.}
\label{tab:soundness-records}
\begin{tabular}{lrr}
\toprule
Reference status & Consistent & Inconsistent \\
\midrule
Admissible & $2\,735$ & $0$ \\
Retained detectable alteration & $0$ & $2\,265$ \\
\bottomrule
\end{tabular}
\end{table}
Localization reproduces every reference label in
Table~\ref{tab:soundness-records}, with no classification
disagreements.
\par
Because the checker determines which alterations are retained,
this measures conditional agreement between reference
construction, checking, and localization.
It does not estimate detection accuracy against independently
established errors or identify the originally altered cell.
The $2\,683$ conflicts exceed the $2\,265$ altered records
because one alteration may violate several requirements.

\subsection{Empirical Verification of Zero-Energy Corrections}
\label{subsec:empirical-correction}

The correction experiment uses the canonical diagnostics to guide
modifications of all $10\,158$ initially inconsistent records. The
donor pool contains the $4\,842$ initially consistent records and
remains fixed throughout the run. Correction uses \texttt{STRATUM},
the five-category territorial partition comprising North-West,
North-East, Centre, South, and Islands, rather than the finer province
grouping used for sample selection.
\par
The record identifier and stratum are immutable, and additional
restrictions protect exogenous auxiliary information. Candidates
respect the prescribed domains and editability restrictions.
Endogenous auxiliary coherence is evaluated against candidate inputs,
and semantic-edit outcomes are recomputed. LEDIT outcomes are
therefore never assigned independently:
$1$ and $2$ are accepted, while \texttt{NA} denotes a violation.
\par
A reference implementation uses local stochastic search within the
dependency scopes of reported violations. Donor distributions are
conditional on stratum and the preceding variable in the adopted
ordering, with progressively relaxed conditioning when usable donors
are insufficient. Derived values and affected checks are reevaluated
for each proposal. 
From an operational perspective, this couples value proposals with
reassessment of routing and semantic requirements, allowing the
correction process to address requirements activated by earlier
changes without first solving an exact minimum-cost localization
problem. 
Search prioritizes violation counts and uses
modification cost to distinguish ties; it does not test the
scalar-weight condition of
Theorem~\ref{thm:hamiltonian-exact-penalty}. 
That condition concerns global minimizers of the combined scalar
objective. Verification of admissibility and certificates based
on attained cost lower bounds do not depend on that condition. 
The search stops when all
diagnostic counts vanish or the sweep limit is reached.
Appendix~\ref{subsec:empirical-search-details} gives the donor,
fallback, and execution details.
\par
Final routing, semantic, and auxiliary-coherence checks verify
returned records independently of search acceptance. They establish
compliance with the declared requirements, rather than agreement with
independently known true values. Table~\ref{tab:empirical-correction}
reports the results.
\begin{table}[htbp]
\centering
\caption{Correction results for the complete validation sample.}
\label{tab:empirical-correction}
\begin{tabular}{lr}
\toprule
Measure & Result \\
\midrule
Initially consistent donor records & $4\,842$ \\
Initially inconsistent records processed & $10\,158$ \\
Records resolved & $10\,146$ \\
Records unresolved & $12$ \\
Violations before correction & $16\,868$ \\
Violations after correction & $12$ \\
Cells changed & $29\,085$ \\
Blank-to-value changes & $24\,218$ \\
Value-to-blank changes & $3\,290$ \\
Replacements of nonblank values & $1\,577$ \\
Reported correction time (seconds) & $1\,394.82$ \\
\bottomrule
\end{tabular}
\end{table}
\par
For each of the $10\,146$ resolved records, all three diagnostic
components vanish. A complementary audit confirms root activity, modal
encoding, and unique predecessor support, with no unsupported present
variables or missing reached mandatory variables. Together with the
correction restrictions, these checks establish admissibility and
hence
$H_{\mathrm{viol}}=0$ by
Proposition~\ref{prop:hamiltonian-zero-penalty}. Each verified
correction therefore attains the global minimum of violation energy
within its permitted candidate domain. The additional assessment of
modification cost follows in
Subsection~\ref{subsec:empirical-minimum-cost}.
\par
All $12$ residual violations concern a required but missing value of
the exogenous variable \texttt{VAUX6}. Their household identifiers
coincide exactly with the $12$ unmatched records in the relevant
external source. The restrictions prevent filling these gaps with
unsupported external values. The diagnostics therefore identify an
unresolved information gap, without proving that every permitted
alternative correction is infeasible.
\par
The $29\,085$ changed cells represent $1.123\%$ of the report's
denominator of $2\,590\,290$ cells across processed records and
columns not marked immutable. This measure includes derived columns in
the extended working representation. It therefore describes
intervention across those columns, rather than changes confined to
original questionnaire items or the charged cost used in the following
audit.
\par
The reported correction time is approximately $23.25$ minutes,
excluding the initial check and donor-cache preparation. It describes
the execution considered here and does not establish relative
computational performance. The experiment verifies the specified
consistency requirements; the statistical quality of replacement
values requires separate assessment.

\subsection{Minimum-Cost Certificates}
\label{subsec:empirical-minimum-cost}

A retrospective audit applies
Equation~\eqref{eq:minimum-cost-certificate}
to the returned corrections.
Costs are computed on $247$ charged coordinates,
with filling, blanking, and replacement penalties
of $1$, $2$, and $3$, respectively.
The charged set excludes the identifier, immutable coordinates,
deterministically recalculated coordinates,
and mode-locked auxiliaries.
\par
Among the resolved records, $493$ have cost equal to $1$.
For $469$ of them, an initially violated canonical implication
involves an activating coordinate and a missing required
coordinate, both charged for modification.
Every admissible correction must change at least one,
at a cost of at least $1$.
This establishes $C^\star(x)\geq1$.
The returned correction fills the required coordinate
and leaves all other charged coordinates unchanged.
Since its admissibility has been verified, it attains
the bound:
$C^\star(x)=C(\widehat y\mid x)=1$.
These $469$ certificates require no rerun of the correction
procedure.
\par
The remaining $24$ cost-one records also pass the final checks
but are not certified by this restricted witness search;
absence of a certificate does not imply suboptimality.
The audit establishes minimum cost for the certified subset.
It neither verifies the sufficient weight condition in
Theorem~\ref{thm:hamiltonian-exact-penalty}
nor establishes cost optimality for all returned corrections.


\section{Discussion}
\label{sec:discussion}

The framework gives statistical production a common specification for
rule generation, traceable diagnosis, and correction verification.
Indexed dependencies link conflicts to their metadata sources and
identify checks affected by changes. Validation uses the generating
requirements without materializing their full deductive closure.
Compatibility and incompatibility formulations defining the same
admissible set retain the same minimum-cost problem. Exact
decomposition requires independence across all relevant
requirements, including semantic and auxiliary dependencies.
With additive costs, independent blocks can be solved separately
without loss of global optimality. By contrast, sequential
correction of blocks with shared variables can exclude globally
cheaper corrections when earlier decisions are fixed. The
framework avoids requiring this separation and retains the
declared dependencies within a common correction specification.
\par
The guarantees depend on complete requirements, domains preserving the
distinctions used by semantic predicates, and appropriate replacement
restrictions. Accepting an unevaluable semantic outcome does not
establish the missing relationship. Positive violation weights
preserve the admissible zero set, while their values can affect search
behaviour. Modification costs express an intervention criterion, not
statistical information loss. Efficient candidate checking does not
imply efficient global optimization.
\par
The application demonstrates implementation agreement and verified
consistency restoration. The lower-bound audit certifies minimum cost
for a subset of outputs; it does not establish optimality for the
remainder or test the sufficient scalar-weight condition. Residual
external-information gaps remain subject to replacement restrictions,
and search failure does not prove infeasibility. The selected survey
extract and checker-filtered synthetic alterations establish neither
population error prevalence nor detection accuracy against independent
truth. The runtime describes one execution and supports no claim of
computational superiority for the search procedure.
\par
Further work should test stronger lower bounds on instances with known
optima and assess sensitivity to weights, costs, and editability.
Independent benchmarks should compare detection and recovery under
common requirements, candidate restrictions, and stopping criteria.
Replacement quality requires assessment of distributions,
associations, and effects on survey estimates.


\section{Conclusion}
\label{sec:conclusion}

This study links routing metadata and semantic edits to traceable
diagnosis and verifiable correction for dependencies admitting the
specified normalized acyclic representation. With complete diagnostics and positive weights, the violation penalty
characterizes admissibility by its zero set within the declared domain. Minimum modification cost is addressed
separately through sufficient penalty conditions and certificates
based on attained lower bounds.
\par
In the household survey application, $10\,146$ of $10\,158$ initially
inconsistent records were corrected and verified as admissible; $12$
retained missing external information subject to replacement
restrictions. Minimum modification cost was certified for $469$
outputs without rerunning search. These results support auditable
consistency restoration. The statistical quality of replacement values
and their effects on survey estimates require separate assessment.


\section*{Declarations}

\subsection*{Conflict of interest}
The author declares that there is no conflict of interest.

\subsection*{Funding}
The author received no specific funding for this research.

\subsection*{Data and code availability}
The code and data required to reproduce the case study are maintained
in a GitHub repository. The repository URL is available from the
author upon reasonable request.
%


\appendix
\numberwithin{equation}{section}
\numberwithin{theorem}{section}

\section{Formal Construction and Theoretical Results}
\label{app:formal}

The construction proceeds from routing metadata to validation rules,
record-level diagnostics, and correction objectives. Following the
questionnaire representation in
\citet{Massoli2008}, normalization resolves compound routing
constructs, including OR and XOR groups, and introduces auxiliary
variables where needed to obtain deterministic routing and acyclicity.
For routing evaluation, numerical values are mapped to the finite
categories prescribed by the specification; values outside these
domains are represented as $\mathrm{NULL}$ in the routing encoding.
The original information needed by semantic predicates must remain
available. The results below concern this normalized representation
and the explicitly supplied requirements.

\subsection{Canonical Representation of the Routing Structure}
\label{subsec:canonical-routing}

Let $G=(V,E)$ be a finite DAG with a distinguished root $V_\rho$,
variables $V=\{V_1,\ldots,V_n\}$, and finite categorical domains
$D_i=\{1,\ldots,k_i\}$, where $k_i\geq1$.
Fix a topological ordering such that $(V_i,V_j)\in E$ implies
$i<j$.
An active variable takes exactly one modality in its domain;
$\mathrm{NULL}$ represents an inactive variable in an admissible
record.
Throughout, \textit{modal} means specific to a modality and does
not refer to modal logic.
The normalized routing specification assigns a unique successor
to each variable--modality pair:
\begin{equation}
\tau:\bigcup_{i=1}^{n}(\{V_i\}\times D_i)
\longrightarrow V\cup\{\mathrm{END}\}.
\label{eq:transition-function}
\end{equation}
Here $\mathrm{END}\notin V$ is the terminal state.
If $\tau(V_i,a)=V_j\in V$, then $(V_i,V_j)\in E$ and hence
$i<j$.
For each pair $(V_i,a)$, let $S(V_i,a)$ be the complete set
of variables specified as structurally inactive by that decision.
The complete transition is
\begin{equation}
T(V_i,a)=\bigl(\tau(V_i,a),S(V_i,a)\bigr).
\label{eq:complete-routing-transition}
\end{equation}
The family $\mathcal T=(T(V_i,a))_{i,a}$ is indexed by source
and modality.
Distinct indices are retained even when they specify the same
successor and inactive-variable set.
We assume that the specification is coherent: every variable
in $S(V_i,a)$ is absent from every root-to-terminal path using
that transition.
For every modal transition, we also require
$\tau(V_i,a)\in V \Rightarrow \tau(V_i,a)\notin S(V_i,a)$,
including when the transition is unreachable from the root.
Starting at the root, successive selected modalities determine
a unique path, which terminates because each nonterminal
transition advances in a finite topological order.
The expanded transposed incidence representation $\mathbf B^T$
retains one indexed row for every modal transition:
\begin{equation}
\begin{aligned}
\operatorname{Rows}(\mathbf B^T)
  &=\{r_{i,a}:V_i\in V,\ a\in D_i\},\\
r_{i,a}
  &=\bigl(V_i,a,\tau(V_i,a),S(V_i,a)\bigr).
\end{aligned}
\label{eq:canonical-modal-row}
\end{equation}
The tuple specifies the information attached to the row,
including its modal identifier.
Rows therefore remain distinct when their numerical incidence
patterns coincide.
Terminal transitions are also retained, although they activate
no further variable.
Each indexed transition can be recovered from its row.
Canonicity is relative to the fixed normalization, domains,
transitions, and topological indexing; it does not assert
uniqueness across different normalizations of a questionnaire.

\subsection{Horn Theory Induced by the Routing Structure}
\label{subsec:horn-theory}

Each canonical row specifies the requirements activated by
one observed modality.
For $r=r_{i,a}$, write $S_r=S(V_i,a)$.
The corresponding compound implication is
\begin{equation}
\varphi_r:\quad V_i=a\Longrightarrow\Psi_r,
\qquad
\Psi_r=
\begin{cases}
V_j\wedge\displaystyle\bigwedge_{V_k\in S_r}\neg V_k,
 &\tau(V_i,a)=V_j\in V,\\[4pt]
\displaystyle\bigwedge_{V_k\in S_r}\neg V_k,
 &\tau(V_i,a)=\mathrm{END}.
\end{cases}
\label{eq:compound-routing-implication}
\end{equation}
Here $V_j$ denotes activity and $\neg V_j$ inactivity,
while $V_i=a$ is a modal atom.
Their interpretations are linked through the record encoding
defined below.
An empty conjunction is $\top$.
Separating the consequent gives the atomic block
\begin{equation}
\begin{aligned}
\mathcal D(r)&=
\begin{cases}
\{V_i=a\Longrightarrow V_j\},&\tau(V_i,a)=V_j\in V,\\
\varnothing,&\tau(V_i,a)=\mathrm{END},
\end{cases}\\
\mathcal S(r)
  &=\{V_i=a\Longrightarrow\neg V_k:V_k\in S_r\},\\
\mathcal R(r)&=\mathcal D(r)\cup\mathcal S(r).
\end{aligned}
\label{eq:atomic-expansion}
\end{equation}
An activation requirement is a definite Horn clause.
An inactivity requirement $V_i=a\Rightarrow\neg V_k$
is equivalently the Horn integrity constraint
$(V_i=a)\wedge V_k\Rightarrow\bot$.
Thus every member of $\mathcal R(r)$ is Horn, and
$\varphi_r\equiv\bigwedge_{h\in\mathcal R(r)}h$.
A terminal row with $S_r=\varnothing$ generates an empty block. 
The routing theory and its indexed block family are
\begin{equation}
\mathcal H=
\bigcup_{r\in\operatorname{Rows}(\mathbf B^T)}\mathcal R(r),
\qquad
\mathfrak R(\mathbf B^T)
 =\bigl(\mathcal R(r)\bigr)_{r\in\operatorname{Rows}(\mathbf B^T)}.
\label{eq:routing-horn-theory}
\end{equation}
The union $\mathcal H$ collects the clauses used for logical
satisfaction, whereas $\mathfrak R(\mathbf B^T)$ retains their
association with routing events.
The generator $\Gamma(r)=\mathcal R(r)$ assigns one block
to each row.
A block may contain several clauses or none.
Repeated clauses are identified in the union, but their indexed
occurrences remain available for diagnostics.
\begin{theorem}[Canonical Horn representation]
\label{thm:canonical-horn-representation}
For a fixed normalized routing specification, $\mathbf B^T$,
$\mathfrak R(\mathbf B^T)$, and $\mathcal H$ are uniquely
determined.
Each complete modal transition corresponds to exactly one
indexed block preserving its activation and inactivity
requirements.
\end{theorem}
\begin{proof}
The construction assigns one row to each indexed transition,
retaining its source, modality, successor, and inactive-variable
set.
Equation~\eqref{eq:atomic-expansion} uniquely determines its
block, whose conjunction is equivalent to the compound routing
implication.
Retaining the indices determines $\mathfrak R(\mathbf B^T)$;
taking the union determines $\mathcal H$.
\end{proof}
The resulting requirements can be evaluated directly from
a record without traversing the original graph.
\par
For a finite propositional Horn theory, \textit{forward chaining}
starts from supplied facts and repeatedly adds atomic conclusions
whose antecedents hold.
At termination, if no integrity constraint has its antecedent
satisfied, the resulting set of true atoms defines the least
model of the theory together with those facts, with respect
to set inclusion \citep{VanEmden1976}.
Otherwise, the theory and facts are jointly unsatisfiable.
Appropriate implementations decide propositional Horn
satisfiability in time linear in the total input size
\citep{Dowling1984}.

\subsection{Routing Admissibility and Logical Completeness}
\label{subsec:routing-admissibility}

To characterize admissibility, the generated clauses must be
interpreted on complete records. An admissible modal routing path
starts at $V_\rho$, follows $\tau$ at each selected modality, and
terminates at $\mathrm{END}$:
\begin{equation}
P=\bigl((V_{i_1},a_1),\ldots,(V_{i_t},a_t),\mathrm{END}\bigr),
\qquad V_{i_1}=V_\rho.
\label{eq:admissible-modal-routing-path}
\end{equation}
Here $a_\ell\in D_{i_\ell}$,
$\tau(V_{i_\ell},a_\ell)=V_{i_{\ell+1}}$ for $\ell<t$,
and $\tau(V_{i_t},a_t)=\mathrm{END}$. Its complete record is
\begin{equation}
\begin{aligned}
x_i(P)&=
\begin{cases}
a_\ell,&i=i_\ell\text{ for some }\ell,\\
\mathrm{NULL},&V_i\text{ is not visited by }P,
\end{cases}\\
\mathcal A_G&=\{x(P):P\text{ is admissible}\}.
\end{aligned}
\label{eq:admissible-configuration-set}
\end{equation}
Each path determines one record, and the path is recovered by
following its recorded modalities from the root.
\begin{definition}[Well-formed modal record]
\label{def:well-formed-modal-record}
A complete record $x\in\prod_i(D_i\cup\{\mathrm{NULL}\})$ is well
formed if the root is active and every active nonroot variable $V_j$
has exactly one active modal predecessor: exactly one pair $(V_i,x_i)$
with $x_i\in D_i$ satisfies
$\tau(V_i,x_i)=V_j$.
The set of such records is denoted by $\mathcal X_G$.
\end{definition}
The modal atom $V_i=a$ is true exactly when $x_i=a$, and the activity
atom $V_i$ is true exactly when $x_i\in D_i$. Atoms are therefore
evaluated from the coordinates rather than assigned independently. The
predecessor condition excludes unsupported active variables; the
generated destination clauses additionally require selected successors
to be active. Identifying each record with its induced interpretation,
define
\begin{equation}
\operatorname{Mod}_G(\mathcal H)
 =\operatorname{Mod}(\mathcal H)\cap\mathcal X_G.
\label{eq:modal-model-set}
\end{equation}
\begin{theorem}[Soundness and logical completeness]
\label{thm:soundness-completeness}
For the normalized deterministic routing specification described
above, the canonical Horn theory characterizes precisely the
admissible configurations within $\mathcal X_G$:
\begin{equation}
\mathcal A_G=\operatorname{Mod}_G(\mathcal H).
\label{eq:exact-logical-characterization}
\end{equation}
\end{theorem}
\begin{proof}
First, let $x=x(P)\in\mathcal A_G$. At each visited modal event, the
nonterminal successor is the next active variable on $P$. Coherence
ensures that every variable in its inactive-variable set is absent
from $P$. The associated block is therefore satisfied. All remaining
blocks have false antecedents, so
$x\models\mathcal H$.
The root is active. Each active nonroot variable has its preceding
variable on $P$ as an active modal predecessor. No other active
variable can select it, because every selected transition on $P$ leads
to the next variable or to $\mathrm{END}$. Hence $x\in\mathcal X_G$.
\par
Conversely, let $x\in\mathcal X_G$ satisfy $\mathcal H$. Starting at
the active root, each recorded modality selects a unique successor,
which the destination clause requires to be active unless it is
$\mathrm{END}$. Since every nonterminal transition increases the
topological index, this constructs a finite path $P$ terminating at
$\mathrm{END}$. Suppose that an active variable is absent from $P$.
Trace its active modal predecessors backwards. Their indices strictly
decrease, so the chain must terminate. Its terminal variable must be
the root: any active nonroot variable would have an active modal
predecessor by Definition~\ref{def:well-formed-modal-record}.
Reversing this chain gives a sequence of selected transitions from the
root to the supposed off-path variable. Each active variable has a
single recorded modality and a unique selected successor, so this
sequence must be an initial segment of $P$, a contradiction. Thus the
active coordinates and their modalities are exactly those of $x(P)$,
and $x\in\mathcal A_G$.
\end{proof}
The restriction to $\mathcal X_G$ is essential. An all-inactive
record, for example, satisfies every routing implication through a
false antecedent but has an inactive root. For unrestricted input
records, domains, root activity, and predecessor support must
therefore be checked alongside the clauses. These conditions must
likewise be enforced or diagnosed during correction. They are part of
the acceptance specification and must remain fixed when assessing
whether a routing clause is redundant.

\subsection{Embedding Semantic Edits into the Canonical Framework}
\label{subsec:semantic-embedding}

Semantic edits express logical or arithmetic requirements
that need not follow from routing.
They are supplied explicitly and incorporated through
deterministically evaluated auxiliary nodes.
The domains $D_i$ must retain distinctions that affect their
outcomes.
If a routing category groups such values, it must be refined
for joint validation, with the refined values retaining
the same routing successor and inactivity requirements.
Numerical codes identify represented values; arithmetic
predicates operate on those values rather than on arbitrary
category labels. 
For edit $q$, let $I_q$ be its input-coordinate scope and let
\begin{equation}
\phi^{(q)}:
\prod_{i\in I_q}(D_i\cup\{\mathrm{NULL}\})
\longrightarrow\{1,2,\mathrm{NULL}\}
\label{eq:semantic-predicate}
\end{equation}
be its deterministic evaluation map.
Under the conditional convention used here, outcome $1$
means that the edit applies and is satisfied; outcome $2$
means that it is inactive or unevaluable because a required
input is unobserved; $\mathrm{NULL}$ means that it applies,
is evaluable, and is violated.
No particular algebraic form is imposed on the predicate.
Let $A_q\subseteq\{1,2\}$ be the nonempty accepted-outcome set,
and define
\begin{equation}
F_q=\bigl(\{1,2\}\setminus A_q\bigr)\cup\{\mathrm{NULL}\}.
\label{eq:forbidden-semantic-outcomes}
\end{equation}
The stated convention uses $A_q=\{1,2\}$.
Accepting outcome $2$ records no semantic violation;
it does not establish the truth of an unevaluable condition.
Missing inputs required by routing remain subject
to the routing checks. 
Introduce an auxiliary node $E^{(q)}$ with evaluated value
$e^{(q)}(x)=\phi^{(q)}((x_i)_{i\in I_q})$.
Here $\mathrm{NULL}$ denotes a rejected semantic outcome,
rather than structural inactivity.
The atom $E^{(q)}=\mathrm{NULL}$ is true exactly when
the evaluation returns that outcome.
For $m$ edits, define
\begin{equation}
x^+=\bigl(x,e^{(1)}(x),\ldots,e^{(m)}(x)\bigr),
\qquad
\mathcal X_G^+=\{x^+:x\in\mathcal X_G\}.
\label{eq:augmented-modal-record}
\end{equation}
Membership in $\mathcal X_G^+$ requires deterministic
augmentation: semantic outcomes cannot be assigned independently
of their inputs. 
The extended DAG $G^+=(V^+,E^+)$ has
$V^+=V\cup\{E^{(1)},\ldots,E^{(m)}\}$ and adds edges
$V_i\to E^{(q)}$ for $i\in I_q$.
These encode functional dependencies rather than routing
successors.
Placing auxiliaries after their inputs preserves acyclicity.
Dependencies on other auxiliaries can be accommodated through
an acyclic evaluation order; composition then expresses each
outcome as a function of its underlying record inputs. 
For each forbidden event $E^{(q)}=b$, introduce a
semantic-validation row $s_{q,b}$:
\begin{equation}
\begin{aligned}
\operatorname{Rows}(\mathbf B^{T,+})
 &=\operatorname{Rows}(\mathbf B^T)
   \mathbin{\dot\cup}\{s_{q,b}:1\leq q\leq m,\ b\in F_q\},\\
\mathcal R^+(s_{q,b})
 &=\{E^{(q)}=b\Longrightarrow\bot\},\\
\mathcal H^+
 &=\bigcup_{u\in\operatorname{Rows}(\mathbf B^{T,+})}
   \mathcal R^+(u),
\end{aligned}
\label{eq:extended-horn-theory}
\end{equation}
where $\mathcal R^+(r)=\mathcal R(r)$ for routing rows.
Semantic rows index rejection events and do not modify $\tau$;
the added dependency edges identify the inputs used
to evaluate them.
The jointly admissible set is
\begin{equation}
\mathcal A^+
 =\{x^+:x\in\mathcal A_G,\
 e^{(q)}(x)\in A_q\text{ for every }q\}.
\label{eq:joint-admissible-set}
\end{equation}
\begin{theorem}[Reduction to extended Horn consistency]
\label{thm:semantic-reduction}
For every deterministic augmentation $x^+\in\mathcal X_G^+$,
\begin{equation}
x^+\models\mathcal H^+
\quad\Longleftrightarrow\quad x^+\in\mathcal A^+.
\label{eq:semantic-reduction}
\end{equation}
\end{theorem}
\begin{proof}
By Theorem~\ref{thm:soundness-completeness}, the routing clauses
hold precisely when $x\in\mathcal A_G$.
For each edit, deterministic evaluation assigns one outcome
in $\{1,2,\mathrm{NULL}\}$.
Its rejection clauses hold precisely when that outcome
belongs to $A_q$.
These are exactly the conditions defining $\mathcal A^+$.
\end{proof}
Consequently,
$\operatorname{Mod}_{G^+}(\mathcal H^+)
:=\operatorname{Mod}(\mathcal H^+)\cap\mathcal X_G^+
=\mathcal A^+$.
The reduction connects evaluated predicates to Horn integrity
constraints; it does not require arbitrary arithmetic
predicates themselves to be Horn formulae.

\subsection{Minimality and Closure-Free Consistency}
\label{subsec:minimality-closure-free}

The indexed basis
\begin{equation}
\mathfrak B^+=\bigl(\mathcal R^+(u)\bigr)_
 {u\in\operatorname{Rows}(\mathbf B^{T,+})}
\label{eq:extended-indexed-basis}
\end{equation}
retains one block for each routing transition or semantic rejection
event. Deleting a block removes its indexed event from the
representation, including provenance. This event-preserving
completeness includes empty terminal blocks and does not imply a
minimum clause count. Logical irredundancy instead asks whether
deleting any individual clause changes the accepted record set, with
the evaluation domain and all additional checks held fixed. Within
$\mathcal X_G$, the destination clauses reconstruct a complete
root-to-terminal path, as shown in
Theorem~\ref{thm:soundness-completeness}. Coherence then implies the
skip requirements. Their explicit retention remains useful for
diagnostics on unrestricted inputs, even though they are redundant
relative to the well-formed domain. Irredundancy relative to $\mathcal
X_G^+$ requires, for every
$h\in\mathcal H^+$,
\begin{equation}
\operatorname{Mod}(\mathcal H^+)\cap\mathcal X_G^+
\neq
\operatorname{Mod}(\mathcal H^+\setminus\{h\})
 \cap\mathcal X_G^+.
\label{eq:relative-logical-irredundancy}
\end{equation}
\begin{proposition}[Witness criterion for logical irredundancy]
\label{prop:witness-irredundancy}
The theory is logically irredundant relative to $\mathcal X_G^+$ if
and only if every $h\in\mathcal H^+$ admits a witness
$x_h^+\in\mathcal X_G^+$ such that
\begin{equation}
x_h^+\models\mathcal H^+\setminus\{h\},
\qquad x_h^+\not\models h.
\label{eq:irredundancy-witness}
\end{equation}
\end{proposition}
\begin{proof}
Deleting a clause can only enlarge the relative model set. The
enlargement is strict exactly when a record in
$\mathcal X_G^+$ satisfies the remaining clauses but violates
the deleted clause.
\end{proof}
A deletion-based inference test establishes this property only if it
supplies such witnesses or otherwise proves their existence in
$\mathcal X_G^+$. Non-derivability under a particular inference
procedure alone is insufficient. The evaluation domain matters. In
Section~\ref{sec:illustrative_questionnaire}, removing
$(Q_1=\mathrm{No})\Rightarrow B_2$ would admit
$(\mathrm{No},1200,1000)$ if only root activity, value domains,
and the remaining displayed routing clauses were checked. It remains
excluded in $\mathcal X_G$ because $Q_2$ has no active predecessor
selecting it. Thus deleting a clause need not change the represented
routing when another condition enforces its requirement. The criterion
above applies to any fixed evaluation domain, but irredundancy on a
larger domain does not imply it on
$\mathcal X_G^+$. No general clause-minimality claim is required
for complete validation. Let $\operatorname{Cn}(\mathcal H^+)
=\{g:\mathcal H^+\models g\}$ denote the deductive closure.
\begin{theorem}[Closure-free consistency]
\label{thm:closure-free-consistency}
For every $x^+\in\mathcal X_G^+$,
\begin{equation}
x^+\in\mathcal A^+
\quad\Longleftrightarrow\quad x^+\models\mathcal H^+
\quad\Longleftrightarrow\quad
x^+\models\operatorname{Cn}(\mathcal H^+).
\label{eq:closure-free-equivalence}
\end{equation}
\end{theorem}
\begin{proof}
The first equivalence is Theorem~\ref{thm:semantic-reduction}. Every
model of a theory satisfies its logical consequences; conversely, the
theory is contained in its deductive closure.
\end{proof}
Thus materializing further consequences cannot exclude a record
already satisfying the generating theory. This does not mean that
$\mathcal H^+=\operatorname{Cn}(\mathcal H^+)$.
A finite operational closure produced by an implementation must also
be distinguished from the full deductive closure. The direct violation
count
\begin{equation}
\nu(x^+)=\sum_{h\in\mathcal H^+}
 \mathbf1\{x^+\not\models h\}
\label{eq:direct-violation-count}
\end{equation}
vanishes exactly on $\mathcal A^+$ within $\mathcal X_G^+$. Counting
indexed occurrences rather than distinct clauses preserves this zero
set, provided all requirements are included. Repeated or derived
clauses can nevertheless change positive counts. The admissible set is
therefore fixed by the logical characterization, whereas the
diagnostic representation determines penalties outside that set.

\subsection{Localization of Data Violations}
\label{subsec:canonical-localization}

For an input record $x\in\prod_{i=1}^n D_i^+$, where
$D_i^+=D_i\cup\{\mathrm{NULL}\}$, define
\begin{equation}
z_{i,a}(x)=\mathbf1\{x_i=a\},
\qquad
\mathbf z(x)=\bigl(z_{i,a}(x)\bigr)_{i,a}^{T},
\qquad
\sum_{a\in D_i^+}z_{i,a}(x)=1.
\label{eq:complete-modal-encoding}
\end{equation}
These indicators are defined throughout the product domain;
well-formedness is an additional condition. 
For a routing row $u=r_{i,a}$, let
$\alpha_u(x)=z_{i,a}(x)$ indicate whether its antecedent holds.
Let $J_u$ contain the indices of its destination, if any,
and of the variables in its inactive-variable set.
For $j\in J_u$, set $e_{u,j}=1$ when activity is required
and $e_{u,j}=0$ when inactivity is required.
With $o_j(x)=\mathbf1\{x_j\neq\mathrm{NULL}\}$, define
\begin{equation}
d_{u,j}(x)=
\begin{cases}
\alpha_u(x)\mathbf1\{o_j(x)\neq e_{u,j}\},&j\in J_u,\\
0,&j\notin J_u.
\end{cases}
\label{eq:row-variable-localization}
\end{equation}
A nonzero indicator identifies either a missing required value
or a present value required to be blank.
An unactivated row contributes no violation.
The diagnostic matrix and routing localization set are
\begin{equation}
\mathbf D(x)=[d_{u,j}(x)],
\qquad
\mathcal L_{\mathrm{route}}(x)=\{(u,j):d_{u,j}(x)=1\}.
\label{eq:routing-localization-set}
\end{equation}
Rows retain the originating events and columns identify
the affected variables.
The sum $\sum_{u,j}d_{u,j}(x)$ counts diagnostic occurrences,
whereas
$\ell_j^{\mathrm{route}}(x)
=\mathbf1\{\sum_u d_{u,j}(x)>0\}$
flags an affected variable.
The computation compares required and observed statuses
directly from canonical rows; explicit expansion into atomic
clauses is unnecessary unless separate clause identifiers
are needed. 
For each semantic edit, define
\begin{equation}
v_q^{\mathrm{sem}}(x)=\mathbf1\{e^{(q)}(x)\in F_q\},
\qquad
\mathcal L_q^{\mathrm{sem}}(x)=
\begin{cases}
I_q,&v_q^{\mathrm{sem}}(x)=1,\\
\varnothing,&v_q^{\mathrm{sem}}(x)=0.
\end{cases}
\label{eq:semantic-edit-violation}
\end{equation}
The scope $I_q$ identifies the inputs of a failed edit.
For derived inputs, their dependencies identify the underlying
coordinates.
The combined output is
\begin{equation}
\begin{aligned}
\mathcal L(x)={}&
 \{(\mathrm{route},u,j):
   (u,j)\in\mathcal L_{\mathrm{route}}(x)\}\\
 &\cup\{(\mathrm{semantic},q,I_q):
   v_q^{\mathrm{sem}}(x)=1\}.
\end{aligned}
\label{eq:complete-canonical-localization}
\end{equation}
Each nonzero indicator identifies a violated generating
requirement, and every such violation is represented.
For $x^+\in\mathcal X_G^+$,
\begin{equation}
\mathcal L(x)=\varnothing
\quad\Longleftrightarrow\quad x^+\models\mathcal H^+
\quad\Longleftrightarrow\quad x^+\in\mathcal A^+.
\label{eq:localization-consistency}
\end{equation}
For unrestricted inputs, domains, root activity, and predecessor
support must additionally be verified.
The diagnostic location need not identify the erroneous cell:
a failure of $V_i=a\Rightarrow V_j$ manifests in $V_j$,
but changing either its value or the activating modality
may remove the violation.

\subsection{From Localization to a Correction Formulation}
\label{subsec:localization-correction}

Let $x$ be the observed record and
$\mathcal L(x)=\{\lambda_1,\ldots,\lambda_s\}$.
For each conflict $\lambda_r$, let $C_r$ contain all editable
input coordinates on which its violated condition depends,
including those determining activation.
Deterministically recomputed coordinates are expanded to their
underlying inputs; independently adjustable coordinates remain
in the scope.
For a direct routing requirement $V_i=a\Rightarrow V_j$
on ordinary variables, $C_r$ contains the editable members
of $\{i,j\}$.
For semantic or derived requirements, it also includes
dependencies determining whether inputs are active or
legitimately blank.
If no coordinate in $C_r$ changes, the violated condition
retains its truth value.
Every admissible correction must therefore have a changed-input
set $S$ satisfying
\begin{equation}
S\cap C_r\neq\varnothing,
\qquad r=1,\ldots,s.
\label{eq:coverage-condition}
\end{equation}
With nonnegative cell costs $c_j$, this gives the weighted
hitting-set problem
\begin{equation}
\min_S\left\{
\sum_{j\in S}c_j:
S\cap C_r\neq\varnothing\text{ for every }r
\right\},
\label{eq:weighted-hitting-set}
\end{equation}
where $S$ ranges over subsets of editable input coordinates.
Equal positive costs give minimum-cardinality coverage.
An empty $C_r$ identifies a conflict that cannot be removed
under the editability restrictions, provided its full
dependency scope has been included.
This formulation describes necessary change sets, not sufficient
conditions for a consistent completion.
Its optimum is a lower bound on a record-level modification
cost only under compatible accounting: uncharged inputs receive
zero coverage cost, and each positive $c_j$ does not exceed
the cost of any permitted change to coordinate $j$.
More generally, every admissible correction must cost at least
the sum of the coverage weights over its changed inputs.
Solving the coverage problem is not required by the energetic
formulation below.
Its role here is to express necessary conditions and support
lower bounds, including simple bounds obtainable from
individual conflicts.
\par
A complete correction formulation instead considers replacement
records.
Let $\mathcal Y(x)$ incorporate permitted values and fixed
fields, and compute each candidate's semantic augmentation
$y^+$ from its inputs.
For a nonnegative modification cost $C(y\mid x)$, the problem is
\begin{equation}
\begin{aligned}
&\min_{y\in\mathcal Y(x)} C(y\mid x)\\
&\text{subject to }y^+\in\mathcal A^+
 \text{ and auxiliary coherence}.
\end{aligned}
\label{eq:admissible-correction-problem}
\end{equation}
The additional coherence condition concerns other auxiliary
coordinates retained in the working representation.
Intermediate search configurations may disagree with their
defining rules; any output certified as an admissible correction
must satisfy them wherever applicable.
Acceptance of an intermediate move does not itself certify
admissibility.
Semantic-edit outcomes are always recomputed as in
Equation~\eqref{eq:augmented-modal-record}.
For example,
$C(y\mid x)=\sum_{j\in J}c_j\mathbf1\{y_j\neq x_j\}$
measures weighted changes over a declared charged set $J$.
The detailed formulation is given in
Appendix~\ref{subsec:minimum-cost-correction}.
A consistent donor establishes feasibility for observation $x$
only if it also belongs to $\mathcal Y(x)$.
Different immutable values or protected external information
may prevent this.
Likewise, failure of a limited search to find a correction
does not establish infeasibility.
\par
Routing dependencies explain why correction must consider
more than the directly modified field.
Changing an upstream modality can require downstream blanking
or activate questions needing additional values.
For a move applied to a current configuration $y$, an effective
operation cost may be written as
\begin{equation}
\operatorname{Cost}(\mu\mid y)
 =c_i(y_i,a')+\operatorname{Cost}_{\mathrm{down}}(\mu\mid y),
\qquad
\mu=(V_i:y_i\leftarrow a').
\label{eq:effective-cost-decomposition}
\end{equation}
The induced term depends on how the move is completed.
This cost is distinct from $C(y\mid x)$, which compares
the final candidate with the original record.
Accumulated move costs need not equal final intervention.
\par
The canonical indices and auxiliary dependencies identify
requirements to re-evaluate after a change, including induced
effects and semantic links between routing branches.
A direct diagnostic scan is linear in the number of inspected
occurrences when each takes bounded work; semantic and
auxiliary predicate-evaluation costs must be added separately.
This bounds the work of checking a candidate, not the number
of candidates needed to find a correction.
It motivates using the diagnostics to evaluate successive
candidate records directly.

\subsection{Hamiltonian Formulation of Data Correction}
\label{subsec:hamiltonian-correction}

The energetic formulation characterizes admissibility through
the zero set of a nonnegative violation objective.
This characterization also provides a scalar target for
correction, independently of the numerical search procedure.
Throughout this and the next subsection, $\mathcal Y(x)$
is a nonempty finite candidate set respecting value domains,
fixed fields, modal encoding, and the root and predecessor
conditions defining $\mathcal X_G$.
Restricting candidates to $\mathcal X_G$ excludes no admissible
correction.
This restriction defines the domain of the results below;
it does not require every intermediate search configuration
to belong to $\mathcal X_G$.
Semantic outcomes are evaluated from candidate inputs.
Other endogenous auxiliaries may be retained as trial
coordinates, with disagreement from their defining rules
diagnosed explicitly.
Externally supplied or mode-locked auxiliaries retain their
prescribed active values and may be blanked only when routing
permits.
These restrictions belong to $\mathcal Y(x)$; missing external
information cannot be freely imputed merely to remove a violation.
\par
Partition the diagnostics into $N_R(y)$ for routing violations
on ordinary variables, $N_L(y)$ for semantic-edit violations,
and $N_A(y)$ for auxiliary-variable violations.
The last component includes routing failures involving
auxiliaries and disagreements with endogenous defining rules.
For a present auxiliary, a disagreement is counted when no
routing failure is already reported for it.
A required but missing auxiliary is diagnosed by routing;
a legitimately inactive auxiliary is not penalized for lacking
a derived value.
The rules and their applicability must ensure that every
failure of required coherence is covered.
This counting convention can affect positive totals but must
not conceal incoherence at zero total count.
Under the convention $A_q=\{1,2\}$,
\begin{equation}
N_L(y)=\sum_{q=1}^{m}
 \mathbf1\{e^{(q)}(y)=\mathrm{NULL}\}.
\label{eq:hamiltonian-semantic-count}
\end{equation}
Routing contributions follow
Equation~\eqref{eq:row-variable-localization}, partitioned
according to whether the affected coordinate is ordinary
or auxiliary.
Counts refer to diagnostic occurrences, so several violations
may involve one variable.
Define
\begin{equation}
\begin{aligned}
N(y)&=N_R(y)+N_L(y)+N_A(y),\\
\mathcal F(x)&=\{y\in\mathcal Y(x):N(y)=0\}.
\end{aligned}
\label{eq:hamiltonian-feasible-set}
\end{equation}
By Theorem~\ref{thm:semantic-reduction} and the coherence checks,
$\mathcal F(x)$ is the feasible set of
Equation~\eqref{eq:admissible-correction-problem}.
For an output evaluated outside $\mathcal X_G$, zero diagnostic
counts certify admissibility only if root and predecessor
failures are also excluded or counted.
For strictly positive weights $w_R,w_L,w_A$, define
\begin{equation}
H_{\mathrm{viol}}(y)
 =w_RN_R(y)+w_LN_L(y)+w_AN_A(y).
\label{eq:violation-hamiltonian}
\end{equation}
The weights determine the relative penalties assigned to
diagnostic classes.
Their magnitudes affect the ordering of inconsistent candidates
but not the set attaining zero.
\begin{proposition}[Zero-energy characterization of admissibility]
\label{prop:hamiltonian-zero-penalty}
Under the stated candidate-domain and diagnostic assumptions,
suppose $w_R,w_L,w_A>0$.
Then $H_{\mathrm{viol}}(y)\geq0$ for every $y\in\mathcal Y(x)$,
and
\begin{equation}
H_{\mathrm{viol}}(y)=0
\quad\Longleftrightarrow\quad y\in\mathcal F(x).
\label{eq:violation-zero-admissibility}
\end{equation}
If $\mathcal F(x)\neq\varnothing$, then
\begin{equation}
\begin{aligned}
\min_{y\in\mathcal Y(x)}H_{\mathrm{viol}}(y)&=0,\\
\operatorname*{arg\,min}_{y\in\mathcal Y(x)}
H_{\mathrm{viol}}(y)&=\mathcal F(x).
\end{aligned}
\label{eq:violation-global-minima}
\end{equation}
\end{proposition}
\begin{proof}
The diagnostic counts are nonnegative integers.
Their positively weighted sum vanishes exactly when all
counts vanish, which is equivalent to feasibility.
If $\mathcal F(x)$ is nonempty, its members attain the lower
bound zero, and every global minimizer must belong to that set.
\end{proof}
A verified zero-energy candidate is therefore an admissible
correction under the restrictions applying to $x$.
All feasible candidates minimize the violation objective,
regardless of how many fields they change.
Minimum intervention requires the additional cost introduced
below.
If $\mathcal F(x)$ is empty, the minimum violation energy
is strictly positive because $\mathcal Y(x)$ is finite
and nonempty.
These statements characterize the objective; they do not
guarantee that a numerical search reaches its minimum.
Strict positivity of the weights is essential:
a zero weight can hide an unresolved diagnostic class,
although the explicit check $N(y)=0$ remains valid. 

\subsection{Minimum-Cost Admissible Correction}
\label{subsec:minimum-cost-correction}

Retain $\mathcal Y(x)$ and $\mathcal F(x)$ as defined above.
Let $J$ be the coordinates charged for modification, and define
\begin{equation}
C(y\mid x)=\sum_{j\in J}c_j(x_j,y_j),
\qquad
c_j(a,b)\geq0,\quad c_j(a,a)=0,
\label{eq:minimum-cost-record-cost}
\end{equation}
with finite coordinate costs.
The set $J$ specifies which changes count as intervention.
Fixed fields need not be charged; deterministically recomputed
coordinates may be excluded when their changes are treated
as consequences of modifying their inputs.
Directly adjustable endogenous auxiliaries may be included.
These accounting choices leave the admissibility requirements
unchanged.
A useful specification distinguishes blanking, filling,
and replacement of a present value.
For finite nonnegative penalties $c_b,c_f,c_v$, take
\begin{equation}
\begin{aligned}
C(y\mid x)=\sum_{j\in J}\bigl[
 &c_b\mathbf1\{x_j\neq\mathrm{NULL},
                 y_j=\mathrm{NULL}\}\\
 &+c_f\mathbf1\{x_j=\mathrm{NULL},
                  y_j\neq\mathrm{NULL}\}\\
 &+c_v\mathbf1\{x_j\neq\mathrm{NULL},
                  y_j\neq\mathrm{NULL},y_j\neq x_j\}
\bigr].
\end{aligned}
\label{eq:hamiltonian-change-cost}
\end{equation}
Unit penalties count modified coordinates in $J$.
Induced changes contribute whenever they affect a charged
coordinate.
The cost is measured relative to $x$ and is generally asymmetric
when $c_b\neq c_f$. 
If $\mathcal F(x)\neq\varnothing$, the minimum-cost problem is
\begin{equation}
C^\star(x)=\min_{y\in\mathcal F(x)}C(y\mid x).
\label{eq:minimum-cost-constrained-problem}
\end{equation}
A minimizer exists because the feasible set is finite,
but need not be unique.
By Proposition~\ref{prop:hamiltonian-zero-penalty},
the equivalent formulation is
\begin{equation}
\min_{y\in\mathcal Y(x)}C(y\mid x)
\quad\text{subject to}\quad H_{\mathrm{viol}}(y)=0.
\label{eq:minimum-change-zero-energy}
\end{equation}
To combine both objectives in a scalar energy, define
\begin{equation}
H_\lambda(y\mid x)
 =H_{\mathrm{viol}}(y)+\lambda C(y\mid x),
\qquad \lambda>0.
\label{eq:correction-hamiltonian}
\end{equation}
Within $\mathcal F(x)$, this reduces to $\lambda C(y\mid x)$.
Over the entire candidate domain, however, a low-cost inconsistent
candidate could be preferred unless the violation penalty
is sufficiently strong.
\begin{theorem}[Exact penalty for minimum-cost correction]
\label{thm:hamiltonian-exact-penalty}
Under the assumptions of
Proposition~\ref{prop:hamiltonian-zero-penalty}, suppose
$\mathcal F(x)\neq\varnothing$.
Let $\delta>0$ satisfy
\begin{equation}
H_{\mathrm{viol}}(y)\geq\delta
\quad\text{for every }
y\in\mathcal Y(x)\setminus\mathcal F(x).
\label{eq:minimum-cost-violation-gap}
\end{equation}
If $y_0\in\mathcal F(x)$ and $\lambda>0$ satisfy
\begin{equation}
\lambda C(y_0\mid x)<\delta,
\label{eq:minimum-cost-penalty-condition}
\end{equation}
then
\begin{equation}
\operatorname*{arg\,min}_{y\in\mathcal Y(x)}
H_\lambda(y\mid x)
=
\operatorname*{arg\,min}_{y\in\mathcal F(x)}
C(y\mid x),
\label{eq:hamiltonian-exact-minima}
\end{equation}
and
\begin{equation}
\min_{y\in\mathcal Y(x)}H_\lambda(y\mid x)
=\lambda C^\star(x).
\label{eq:minimum-cost-optimal-energy}
\end{equation}
\end{theorem}
\begin{proof}
For any infeasible candidate,
$H_\lambda(y\mid x)\geq H_{\mathrm{viol}}(y)\geq\delta$,
by nonnegativity of the modification cost.
The feasible candidate $y_0$ instead satisfies
$H_\lambda(y_0\mid x)=\lambda C(y_0\mid x)<\delta$.
No infeasible candidate can therefore be a global minimizer.
On $\mathcal F(x)$, the objective equals $\lambda C(y\mid x)$,
and $\lambda>0$ preserves the cost minimizers.
\end{proof}
For Equation~\eqref{eq:violation-hamiltonian}, one may take
$\delta=w_{\min}:=\min\{w_R,w_L,w_A\}$.
If $\overline C(x)$ is a finite upper bound on modification
cost over $\mathcal Y(x)$, then
\begin{equation}
\lambda\overline C(x)<w_{\min}
\label{eq:hamiltonian-weight-condition}
\end{equation}
is sufficient whenever the feasible set is nonempty.
For Equation~\eqref{eq:hamiltonian-change-cost}, one may use
$\overline C(x)=|J|\max\{c_b,c_f,c_v\}$.
The cost of an available feasible correction can give
a less restrictive condition.
These conditions are sufficient, not necessary. 
The minimum total energy can be positive even when its
minimizers are admissible, because modification has a cost.
For any $\lambda>0$,
\begin{equation}
H_\lambda(y\mid x)=0
\quad\Longleftrightarrow\quad
y\in\mathcal F(x)\ \text{and}\ C(y\mid x)=0.
\label{eq:hamiltonian-zero-energy}
\end{equation}
Alternatively, lexicographic minimization of
$H_{\mathrm{viol}}$ first and $C$ second gives the constrained
minimizers whenever $\mathcal F(x)\neq\varnothing$,
without a scalar weight condition.
Neither formulation removes the search difficulty of finding
a global minimizer or creates feasibility when none exists. 
An individual output can nevertheless be certified without
a general optimality guarantee for the search.
If $L(x)\leq C(y\mid x)$ for every $y\in\mathcal F(x)$,
then any admissible correction $\widehat y$ satisfies
\begin{equation}
L(x)\leq C^\star(x)\leq C(\widehat y\mid x).
\label{eq:minimum-cost-certificate}
\end{equation}
Equality $C(\widehat y\mid x)=L(x)$ certifies minimum cost.
The coverage conditions in
Appendix~\ref{subsec:localization-correction}
provide one source of bounds under the stated accounting
conditions. 
As an example, consider an initially violated implication
$V_i=a\Rightarrow V_j$ involving two directly represented
coordinates charged for modification.
Every admissible correction must change at least one of them.
If every permitted change to either coordinate costs at least
$\eta>0$, then $C^\star(x)\geq\eta$ whenever a feasible
correction exists.
An admissible correction of cost $\eta$ attains this bound
and is minimum-cost.
This certificate requires verification of the returned record
and the lower bound, but no convergence guarantee for
the procedure that produced it.

\section{Additional Implementation Checks}
\label{app:checks}

\subsection{Deletion Tests and Operational Closure}
\label{subsec:empirical-minimality}

The implementation retains the final variable as a column
but does not allocate outgoing modal rows to it.
These implicit terminal events correspond to transitions
to $\mathrm{END}$ with empty inactive-variable sets in
Appendix~\ref{subsec:canonical-routing}.
Their clause blocks are empty, so adding their indices
does not change the validation requirements.
The reported $572$ rows therefore count stored nonterminal
transitions only.
\par
The variable-domain and fixed-width layout files each
contain $255$ entries aligned with the routing representation.
\par
Adding two logical-edit nodes gives a stored representation
with $257$ vertices, $334$ edges, and $575$ modal rows,
while preserving acyclicity and data alignment.
These implementation counts differ in meaning from
the functional-dependency edges and forbidden-outcome rows
of Appendix~\ref{subsec:semantic-embedding}:
a stored modal row need not be a semantic rejection row,
and the stored edges do not enumerate all predicate-input
dependencies.
\par
Each of the $1\,008$ elementary implications was removed in turn.
Forward inference from its antecedent using the remaining implications
did not recover its conclusion in any case. Technical value-to-filled
bridges were excluded. Thus all $1\,008$ deletion tests returned
non-derivability under the adopted procedure. These tests concern the
operational encoding and do not establish logical irredundancy
relative to well-formed records. In particular, they do not encode the
record-domain conditions used to derive skip requirements in
Appendix~\ref{subsec:minimality-closure-free}.
Proposition~\ref{prop:witness-irredundancy} requires a witness that
also satisfies the fixed evaluation-domain conditions. These tests
therefore do not prove that deleting each clause changes the set of
accepted complete records. For the closure comparison, $572$ technical
implications linking observed modal values to variable activity were
added, giving an operational basis of $1\,580$ implications.
Materializing reachability under the adopted forward-inference
procedure returned the same $1\,580$ implications, with no additional
implications. The basis and its materialized operational closure
produced identical consequences for all $5\,977$ tested fact sets:
$977$ exhaustive singleton seeds and $5\,000$ random
multiple-seed sets, generated with seed $123$ and a maximum size of
eight. This result reflects the encoding: elementary antecedents are
modal value literals, whereas conclusions are filled or blank literals
and do not generate further modal values. The experiment verifies
agreement between the two inference implementations on the tested
inputs. Its finite operational closure is distinct from
$\operatorname{Cn}(\mathcal H^+)$; the general validation
result is established by Theorem~\ref{thm:closure-free-consistency}.

\subsection{Size of the Admissible Routing-Pattern Space}
\label{subsec:empirical-patterns}

The canonical representation also supports counting routing
patterns without enumerating complete records.
Let $m_i$ be the number of encoded modalities used by
the counting procedure for routing variable $i$.
Allowing these modalities and a blank state gives
the unrestricted reference space
$N_{\mathrm{total}}=\prod_{i=1}^{255}(m_i+1)$.
The admissible count is obtained by propagating filled
and blank requirements, with the source variable required
to be filled.
Arbitrary-precision integer arithmetic and memoization yield
\begin{equation}
\begin{aligned}
N_{\mathrm{total}} &\simeq 2.813385\times 10^{118},\\
N_{\mathrm{adm}}   &\simeq 2.999369\times 10^{52},\\
N_{\mathrm{adm}}/N_{\mathrm{total}}
                  &\simeq 1.066107\times 10^{-66}.
\end{aligned}
\label{eq:empirical-pattern-counts}
\end{equation}
The calculation visits $474$ memoized states.
These counts concern the implemented modal encoding
and activation conventions before adding the two semantic
edits.
They do not enumerate all quantitative values or all
combinations of the original $332$ columns.
Routing categories may, in particular, group values
distinguished by semantic predicates.
The ratio describes the combinatorial restriction imposed
by routing, not the probability of observing
a consistent household.
A supplementary structural analysis identifies four computational
blocks. This decomposition provides structural information for
alternative correction strategies; it is not required by the proposed
framework, which evaluates the canonical representation without
partitioning the editing problem.

\subsection{Operational Details of the Correction Experiment}
\label{subsec:empirical-search-details}

The donor pool is fixed at the $4\,842$ initially consistent records.
Donor-based filling and replacement use empirical distributions for
each target variable, conditional on the immediately preceding
variable in the adopted variable ordering and on \texttt{STRATUM}.
This preceding variable need not be the routing predecessor selecting
the target through an incoming edge. Its current value, including
\texttt{NA}, identifies the conditioning group among donors. A missing
preceding value is therefore retained as a distinct conditioning state
rather than automatically triggering a fallback. Only donors with
nonmissing target values in the permitted domain contribute to
proposals.
\par
If the conditioning group contains insufficient usable donors, the
procedure first drops the stratum condition while retaining the
preceding variable's value. If this group is also insufficient, it
uses the target variable's marginal donor distribution when permitted.
These distributions supply candidate values, not guarantees of
admissibility. Affected derived values and routing requirements are
reevaluated for every candidate.
\par
Search prioritizes violation counts and uses modification cost to
distinguish tied candidates. It stops when all diagnostic counts
vanish or the maximum number of sweeps is reached. By default, this
maximum equals the number of DAG nodes; each sweep processes the
current violation queue and may evaluate multiple proposals. Execution
uses six workers and $102$ batches of at most $100$ records, without
reusing previous checkpoints. The reported correction time excludes
the initial check and donor-cache preparation. Final checks assess
routing, semantic, and auxiliary coherence independently of search
acceptance.
%
%
%

%
\end{document}